\documentclass[letterpaper]{article} 
\usepackage{aaai2026}  
\nocopyright
\usepackage{times}  
\usepackage{helvet}  
\usepackage{courier}  
\usepackage[hyphens]{url}  
\usepackage{graphicx} 
\usepackage{natbib}  
\usepackage{caption} 
\usepackage{algorithm}
\usepackage{algorithmic}

\usepackage{newfloat}
\usepackage{listings}
\DeclareCaptionStyle{ruled}{labelfont=normalfont,labelsep=colon,strut=off} 
\floatstyle{ruled}
\newfloat{listing}{tb}{lst}{}
\floatname{listing}{Listing}
\newtheorem{theorem}{Theorem}
\newtheorem{lemma}{Lemma}
\newtheorem{definition}
{Definition}

\newenvironment{proof}{\begin{trivlist}\item\noindent{\sc Proof.}}{\hfill$\Box\hspace{2mm}$\end{trivlist}}

\newenvironment{proof-of-claim}{\begin{trivlist}\item\noindent{\sc Proof of Claim.}}{\hfill $\boxtimes\hspace{2mm}$\end{trivlist}}
\usepackage{amssymb,amsmath}
\usepackage{booktabs}

\title{Diffusion of Responsibility in Collective Decision Making}

\author {
    Pavel Naumov,\textsuperscript{\rm 1}
    Jia Tao\textsuperscript{\rm 2}
}
\affiliations {
    \textsuperscript{\rm 1} University of Southampton, United Kingdom\\
     \textsuperscript{\rm 2} Lafayette College, United States\\
    p.naumov@soton.ac.uk,
    taoj@lafayette.edu  
    }

\begin{document}

\maketitle

\begin{abstract}
The term ``diffusion of responsibility'' refers to situations in which multiple agents share responsibility for an outcome, obscuring individual accountability. This paper examines this frequently undesirable phenomenon in the context of collective decision-making mechanisms.

The work shows that if a decision is made by two agents, then the only way to avoid diffusion of responsibility is for one agent to act as a ``dictator'', making the decision unilaterally. In scenarios with more than two agents, any diffusion-free mechanism is an ``elected dictatorship'' where the agents elect a single agent to make a unilateral decision.

The technical results are obtained by defining a bisimulation of decision-making mechanisms, proving that bisimulation preserves responsibility-related properties, and establishing the results for a smallest bisimular mechanism.
\end{abstract}

\section{Introduction}
Autonomous agents impact our daily lives by playing an increasingly significant role in making critical decisions. The responsibility for the outcome of such decisions is often diffused between multiple agents. For instance, in a two-car collision, the responsibility can be shared by the autonomous systems of both cars. Such situations are often undesirable because they create a ``circle of blame'' between developers and owners of the two vehicles. Diffusion of responsibility has been widely studied across social sciences~\cite{ms75pspb,fzg02pspb,llw22ichssr}, law~\cite{i20ijcr,rkfc22jrcd}, ethics~\cite{bb22aie}, and neuroscience~\cite{fdlglk16hbm}.
To ensure broad acceptance of autonomous systems by society and to protect the interests of the involved parties, it is essential to clearly define individual accountability of humans and machines for outcomes of collective decision-making by minimizing the diffusion of responsibility.
In this paper, we study the possibility of designing decision-making mechanisms that completely avoid such diffusion. 

As an example, consider a well-known collective decision-making mechanism devised by the framers of the US Constitution over two centuries ago, see Figure~\ref{constitution figure}. 
\begin{figure}[ht]
\begin{center}
\scalebox{0.5}{\includegraphics{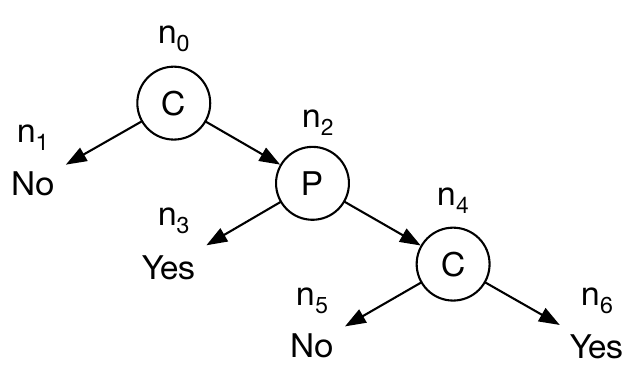}}
\caption{US Constitution mechanism}\label{constitution figure}
\end{center}
\end{figure}
If a bill is rejected by Congress (C), it is dead. If Congress approves a bill, then it is sent to the President (P), who can either sign or veto the bill. Once the bill is signed by the President, it becomes a law of the United States. If the President vetoes the bill, it goes back to Congress. At that point, Congress can either override the veto or kill the bill.

Imagine that a bill to increase funding for AI research is introduced in Congress. Congress passes the bill and sends it to the President. The President decides to veto the bill. When the bill comes back to Congress, it fails to override the veto, and the bill dies. This scenario corresponds to the decision path $n_0,n_2,n_4,n_5$
in Figure~\ref{constitution figure}. Who in this situation, Congress or the President, is responsible for the failure to pass the bill?

Responsibility is a broad term extensively studied in philosophy and law. In philosophy, the focus is on {\em moral} responsibility; in law -- on {\em culpability} or legal responsibility. Although multiple attempts to capture the concept of responsibility have been made, the definition most commonly  cited in philosophy~\cite{w17} is based on Frankfurts' [\citeyear{f69tjop}] principle of alternative possibilities:
\begin{quote}{\em 
\dots\ a person is morally responsible for what he has done only if he could have done otherwise}.
\end{quote}
This principle, sometimes referred to as ``counterfactual possibility''~\cite{c15cop}, is also used to define causality~\cite{lewis13,h16,bs18aaai}.
In this paper, we refer to the responsibility defined based on Frankfurts' principle as {\em counterfactual responsibility} or just {\em responsibility}.
Following recent works in AI~\cite{nt19aaai,ydjal19aamas,bfm21ijcai,s24aaai}, we interpret ``could have done otherwise'' as an agent {\em having a strategy to prevent the outcome no matter what the actions of the other agents are}. 

Let us go back to the decision path in our example ending at node $n_5$. Because the President had a strategy (``sign'') to make the bill into law, the President is counterfactually responsible for the failure to pass the bill. At the same time, Congress also had a strategy to make the bill into law by overriding the presidential veto. Hence, Congress is also counterfactually responsible for the failure to pass the bill. If a decision-making mechanism allows situations when more than one agent is responsible at the same leaf node, then we say that the mechanism admits the {\em diffusion of responsibility}. 
The diffusion of responsibility might lead to a bystander effect or ``circle of blame''. As discussed earlier, in general, the diffusion of responsibility is an undesirable property of a decision-making mechanism. 

Of course, the decision-making mechanism depicted in Figure~\ref{constitution figure} can be modified to eliminate the diffusion of responsibility. In fact, this was the case with the Articles of Confederation, which preceded the US Constitution. The Articles did not establish any sort of national executive, allowing the Congress of the Confederation to pass the legislation without a threat of a veto~\cite{w87psq}, see Figure~\ref{continental figure}. 

\begin{figure}[ht]
\begin{center}
\scalebox{0.5}{\includegraphics{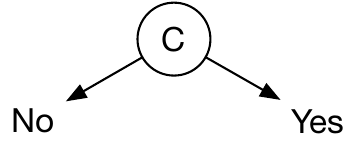}}
\caption{Articles of Confederation}\label{continental figure}
\end{center}
\end{figure}
However, Figure~\ref{continental figure} has an obvious problem. It makes Congress a {dictator} by giving it unlimited power. By a {\em dictator} in a decision-making mechanism, we mean an agent who has an upfront strategy to achieve each possible decision, no matter how the other agents act. It is worth pointing out that Congress is also a dictator in the mechanism depicted in Figure~\ref{constitution figure}. Indeed, to reject the bill, Congress can simply reject it outright (by moving to node $n_1$). To pass the bill, Congress needs to send it to the President (by moving to $n_2$) and, if the President vetoes the bill (by moving to $n_4$), Congress can simply overwrite the veto (by moving to $n_6$).

\begin{figure}[ht]
\begin{center}
\scalebox{0.5}{\includegraphics{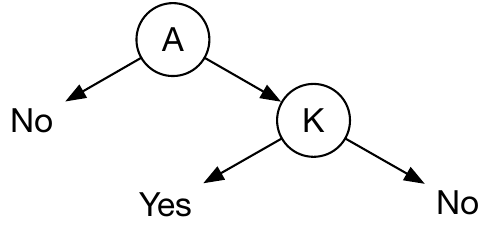}}
\caption{Maryland, Rhode Island, and Connecticut}\label{colony figure}
\end{center}
\vspace{-4mm}
\end{figure}
Are there decision-making mechanisms that do not have a dictator? Yes, there are. In fact, such mechanisms have been used by the 13 colonies that preceded the Confederation. Three of them used the mechanism depicted in Figure~\ref{colony figure} and the others the one shown in Figure~\ref{colony w gov figure}~\cite{w87psq}. Under the mechanism in Figure~\ref{colony figure}, the English King (K) had the absolute power to veto the bills passed by the legislative body of a colony, called the Assembly (A). 
\begin{figure}[ht]
\begin{center}
\scalebox{0.5}{\includegraphics{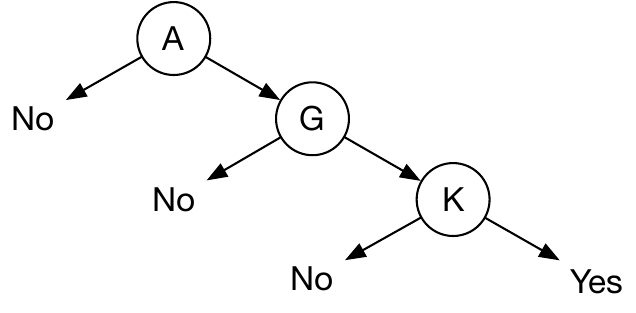}}
\caption{The other 10 colonies}\label{colony w gov figure}
\end{center}
\vspace{-4mm}
\end{figure}
Under the mechanism in Figure~\ref{colony w gov figure}, the Governor of the colony, appointed by the King, also had absolute veto power. Note that all involved parties had an upfront strategy to guarantee that the bill would not become law, but none of them had an upfront strategy to make it into law. 

The mechanism depicted in Figure~\ref{colony figure} has been in effect in England for many centuries. In practice, however, the last English monarch who vetoed (withheld ``royal assent'') a bill passed by the English Parliament was Queen Anne in 1707~\cite{w87psq}. While English monarchs did not withhold royal assent in England, they actively did this in the colonies. This fact was {\em the first} on the long list of complaints against the King listed in the American Declaration of Independence: ``He has refused his Assent to Laws, the most wholesome and necessary for the common good." 

Although the mechanisms shown in Figure~\ref{colony figure} and Figure~\ref{colony w gov figure} do not have a dictator, they both admit the diffusion of responsibility. In the US Constitution mechanism, shown in Figure~\ref{constitution figure}, the diffusion happens when a bill is rejected by Congress after it was vetoed by the President. In the mechanisms in Figure~\ref{colony figure} and Figure~\ref{colony w gov figure}, the diffusion happens when the bill becomes the law because all involved parties had a strategy to prevent this. It is probably worth mentioning that what an agent is responsible for does not have to be negative. If it is a negative thing, the responsible agent is blameworthy; if it is positive, then the agent is praiseworthy.

So far, we have examined mechanisms that (i) allow diffusion and have a dictator, (ii) allow diffusion and have no dictator, and (iii) have a dictator and do not allow diffusion. Are there mechanisms that eliminate diffusion without introducing a dictatorship? Yes, there are. For example, a majority vote by a paper ballot. However, in this paper, we focus on the mechanisms where the agents act in a consecutive order and each of them has perfect information about the previous actions, just like those in Figures~\ref{constitution figure},~\ref{continental figure}, \ref{colony figure}, and \ref{colony w gov figure}. For such mechanisms, the answer to the above question depends on the number of agents. In Theorem~\ref{two players main theorem}, we prove that {\bf any two-agent (consecutive) decision-making mechanism that does not allow diffusion of responsibility must be a dictatorship}. For the case with three agents, Figure~\ref{elected figure} shows one of many possible examples of a decision-making mechanism that has no dictator and is diffusion-free. In this hypothetical example, the Assembly, after passing a bill, can choose to send it to either the Governor or the King. If the Governor is OK with the bill, the Governor returns it back to the Assembly for the final round of voting. If the Governor has concerns about the bill, the bill is sent to the King. In both cases, the King's decision is final.

\begin{figure}[ht]
\begin{center}
\vspace{-2mm}
\scalebox{0.5}{\includegraphics{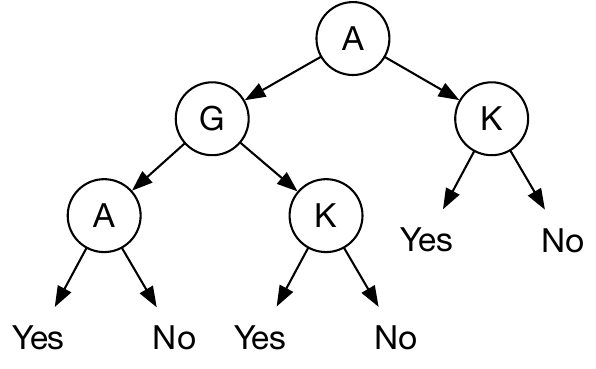}}
\caption{Hypothetical mechanism}\label{elected figure}
\end{center}
\vspace{-4mm}
\end{figure}

One might argue that, in most instances, the diffusion of responsibility and dictatorship are undesirable properties in collective decision-making. Figure~\ref{elected figure} shows that they both can be avoided. At the same time, the mechanism shown in this figure is what we call an {\em elected dictatorship}. 
Intuitively, an elected dictatorship is a mechanism that selects a single agent to be the ``decision-maker'' and this agent decides between all available alternatives. 
In the example depicted in Figure~\ref{elected figure}, the election of a dictator is completed when the decision path reaches pre-leaf nodes (nodes of height 1). Elected dictators in these nodes are agents A, K, and K.

The mechanisms depicted in Figure~\ref{constitution figure}, Figure~\ref{colony figure}, and Figure~\ref{colony w gov figure} are {\em not} elected dictatorships, but they all allow diffusion of responsibility. 
In Theorem~\ref{first main result}, we prove our main result: {\bf any decision-making mechanism that does not allow diffusion of responsibility must be an elected dictatorship}. The proof of this theorem, perhaps unexpectedly, turns out to be highly non-trivial. It uses the bisimulation technique. 
Theorem~\ref{two players main theorem} is a much simpler observation that can be proven independently or, as we have chosen to do, derived from Theorem~\ref{first main result}.  The rest of the paper is structured as follows. First, we formally define a decision-making mechanism and the related notions of responsibility, diffusion, dictatorship, and elected dictatorship. Then, we prove our two results. The last section concludes the presentation.

\section{Decision-Making Mechanisms}

Motivated by the introductory examples, we only consider {\em perfect information} mechanisms in which each agent knows about the actions taken by the previous agents. We also do not consider mechanisms, like ballot voting,  where the agents make their choices ``concurrently''. Finally, we assume that the decision produced by the mechanism is deterministic and does not depend on an ``initial state''.

\begin{definition}\label{game}
A mechanism is a tuple $(Ag,Alt,T,act,\ell)$, where
\begin{enumerate}
    \item $Ag$ is a set of ``agents'',
    \item $Alt$ is a set of ``alternatives'' such that $|Alt|\ge 2$,
    \item $T$ is a finite rooted tree, 
    \item $act$ is a labeling function that maps each non-leaf node $n$ of tree $T$ to an agent $act(n)\in Ag$,
    \item $\ell$ is a surjective labeling function that maps each leaf node $n$ of tree $T$ to an alternative $\ell(n)\in Alt$.
\end{enumerate}
\end{definition}
By $\ell^{-1}(Y)$ we denote the set of all leaf nodes labeled with an alternative $Y\in Alt$. 

In the US Constitution mechanism depicted in Figure~\ref{constitution figure}, the set $Ag$ could be any set containing agents $P$ and $C$. Set $Alt$ consists of two possible alternatives: Yes and No. The value of function $act$ for each non-leaf node is shown inside the circle representing the node. The value of function $\ell$ (either Yes or No) for each leaf node is shown at the leaf node.  In Definition~\ref{game}, we require function $\ell$ to be surjective to avoid ``fictitious'' alternatives that can never be chosen by the mechanism. The elimination of such alternatives is not significant, but it simplifies some of our definitions. Additionally, it makes the following observation true.

\begin{lemma}\label{no singleton mechanisms}
A mechanism contains at least two leaf nodes.     
\end{lemma}
\begin{proof}
If a mechanism contains fewer than two leaf nodes, then the range of leaf-labeling function $\ell$ has a size less than two. Hence, the set $|Alt|<2$ because function $\ell$ is surjective by item~5 of Definition~\ref{game}. The latter contradicts item~2 of the same definition.    
\end{proof}


The key to our definitions of counterfactual responsibility, dictatorship, and elected dictatorship is the notion of the strategy of an agent to achieve a specific set of alternatives. Instead of defining a strategy explicitly, it is more convenient to define a set $win_a(S)$ of ``winning'' nodes from which an agent $a$ has a strategy to guarantee that the alternative chosen by the mechanism belongs to the set $S$. As is common in game theory, we define this set by backward induction:

\begin{definition}\label{df:winning set}
For any set $S$ of alternatives, the set $win_a(S)$ is the minimal set of nodes such that
\begin{enumerate}
\item $\bigcup \ell^{-1}(S)\subseteq win_a(S)$,
\item for any non-leaf node $n$ such that $act(n)=a$, if {\bf\em at least one} child of node $n$ belongs to the set $win_a(S)$, then node $n\in win_a(S)$,
\item for any non-leaf node $n$ such that $act(n)\neq a$, if {\bf\em all} children of node $n$ belong to the set $win_a(S)$, then node $n\in win_a(S)$.
\end{enumerate}
\end{definition}

In our US Constitution mechanism example in Figure~\ref{constitution figure}, the set $win_C(\{\text{Yes}\})$ consists of nodes $n_0$, $n_2$, $n_4$, and $n_6$. Also, $win_C(\{\text{No}\})=\{n_0,n_1,n_4,n_5\}$.


By an {\em ancestor} of a node, we mean any node on the simple path that connects the node with the root of the tree. Ancestors include the node itself and the root. By $Anc(n)$ we denote the set of all ancestors of a node $n$. 
For instance, $Anc(n_5)=\{n_5,n_4,n_2,n_0\}$, see Figure~\ref{constitution figure}.  By a {\em subtree rooted at a node} $n$ we mean the set of all nodes for whom $n$ is an ancestor. Thus, the subtree includes the node $n$ itself.

Following the standard practice in mathematics, we use ``if'' instead of ``iff'' in the next definition and the rest of {\em definitions} (not lemmas or theorems) in this paper.
\begin{definition}\label{responsible}
Agent $a$ is (counterfactually) {\bf\em responsible} at leaf $n$ if $Anc(n)\cap win_a(Alt\setminus \{\ell(n)\})\neq\varnothing$. 
\end{definition}

Note that $\ell(n_5)=\text{No}$, see Figure~\ref{constitution figure}. Thus, $$win_C(Alt\setminus \{\ell(n_5)\})=win_C(\{\text{Yes}\})=
\{n_0,n_2,n_4,n_6\}.$$ At the same time, $Anc(n_5)=\{n_5,n_4,n_2,n_0\}$. Hence, agent $C$ is responsible at the leaf node $n_5$. It is easy to verify that agent $C$ is also responsible at all leaf nodes in this mechanism. Agent $P$ is only responsible at node $n_5$.

\begin{definition}\label{diffusion}
A mechanism allows {\bf\em diffusion of responsibility} if it has a leaf where at least two agents are counterfactually responsible.    
\end{definition}
In our example in Figure~\ref{constitution figure}, such a leaf is $n_5$. A mechanism is {\bf diffusion-free} if it does not allow diffusion of responsibility.

\begin{definition}\label{dict at node}
An agent $a$ is a {\bf\em dictator at node $n$} if node $n$ belongs to $win_a(\{Y\})$ for each alternative $Y\in Alt$. 
\end{definition}
In our example in Figure~\ref{constitution figure}, agent $C$ is a dictator at nodes $n_0$ and $n_4$. Agent $P$ is not a dictator at any of the nodes. 
In Figure~\ref{elected figure},
agents $A$, $K$, and $K$ are the dictators in pre-leaf nodes (nodes of height 1). 

\begin{definition}\label{dictator}
An agent $a$ is a {\bf\em dictator} if the agent is a dictator at the root node.
\end{definition}
Agent $C$ is a dictator in the mechanism depicted in Figure~\ref{constitution figure}. 

\begin{definition}\label{elected dictatorship}
A mechanism is an {\bf\em elected dictatorship} if, for each leaf node $n$, there is a dictator at an ancestor node of~$n$.    
\end{definition}
The mechanism depicted in Figure~\ref{elected figure} (as well as the one in Figure~\ref{constitution figure}) is an example of an elected dictatorship.

\section{Technical results}

In this section, we establish that any decision-making mechanism that does not allow diffusion is, by necessity, an elected dictatorship. Additionally, we show that when there are only two agents, such a mechanism must be a dictatorship. These results are formally presented as Theorem~\ref{first main result} and Theorem~\ref{two players main theorem} at the end of this section.

To prove the first theorem, we introduce the notion of {\em bisimulation} of two decision-making mechanisms. Then, we show that bisimulation preserves the core properties of mechanisms: responsibility, diffusion, and elected dictatorship. Following this, we define a {\em canonical form} of a decision-making mechanism as a smallest mechanism totally bisimular to the given one. Finally, we prove Theorem~\ref{first main result} for the mechanisms in canonical form. Theorem~\ref{two players main theorem} follows from Theorem~\ref{first main result}.

\subsection{Bisimulation}

Before introducing the notion of bisimulation in Definition~\ref{def:bisimulation}, we specify the terminology used in this definition.

\begin{definition}
A directed edge from a parent to a child is labeled with an agent $a$ if the parent node is labeled with $a$.    
\end{definition}
In Figure~\ref{constitution figure}, the directed edge from node $n_2$ to node $n_4$ is labeled with agent $P$. The one from node $n_4$ to node $n_5$ is labeled with agent $C$. 

A {\em trivial} path is a path consisting of a single node.
\begin{definition}\label{df:to}
$n\stackrel{a}{\to}m$ if there is a directed (possibly trivial) path from node $n$ to node $m$ and each edge along the path is labeled with agent $a$.
\end{definition}
As an example, $n_0\stackrel{C}{\to}n_2$ and $n_2\stackrel{C}{\to}n_2$ in Figure~\ref{constitution figure}. 
Note that the relation $n\stackrel{a}{\to}m$ is not trivial because Definition~\ref{game} does {\em not} assume that consecutive nodes have different labels. 
The next lemma follows from Definition~\ref{df:to}.
\begin{lemma}\label{no go lemma}
Suppose $a$ is an agent and node $n$ is such that either it is a leaf node or $act(n)\neq a$. Then, $n\stackrel{a}{\to}m$ implies $n=m$.
\end{lemma}

The next lemma follows from Definition~\ref{df:winning set}.
\begin{lemma}\label{up down lemma}
For any set $S\subseteq Alt$, any two distinct agents $a,b\in Ag$, and any two nodes $n,m$ such that $n\stackrel{a}{\to} m$,
\begin{enumerate}
    \item if $m\in win_a(S)$, then $n\in win_a(S)$,\label{up lemma}
    \item if $n\in win_b(S)$, then $m\in win_b(S)$.\label{down lemma}
\end{enumerate}
\end{lemma}
In Figure~\ref{constitution figure} example, $n_0\in win_C(\{\text{No}\})$ because $n_0\stackrel{C}{\to}n_1$ and $n_1\in win_C(\{\text{No}\})$.
Also, $n_4\in win_C(\{\text{Yes}\})$ because $n_2\stackrel{P}{\to}n_4$ and $n_2\in win_C(\{\text{Yes}\})$. 

The concept of bisimulation of two transition systems is well-known in the literature~\cite{s11}. Intuitively, two systems are bisimular if they exhibit the same ``behavior''. What exactly this means depends on the intended application. In the definition below we define bisimulation in such a way that it preserves core properties of counterfactual responsibility.

\begin{definition}\label{def:bisimulation}
A {\bf\em bisimulation} $R$ of mechanisms $(Ag,Alt,$ $T,act,\ell)$ and $(Ag,Alt,T',act',\ell')$ is a relation between nodes of trees $T$ and $T'$ such that, 
\begin{enumerate}

    \item 
    for any leaf nodes $n$ and $n'$ of trees $T$ and $T'$, respectively, 
    if $nRn'$,
    then $\ell(n)=\ell'(n')$,

    \item 
for any nodes $n,m$ of tree $T$ and any node $n'$ of tree $T'$,
    if $n\stackrel{a}{\to}m$ and $nRn'$, then there is a node $m'$ of tree $T'$ such that $n'\stackrel{a}{\to}m'$ and $mRm'$,

    \item  
for any node $n$ of tree $T$ and any node $n',m'$ of tree $T'$,
    if $n'\stackrel{a}{\to}m'$ and $nRn'$, then there is a node $m$ of tree $T$ such that $n\stackrel{a}{\to}m$ and $mRm'$,
    
    \item for any nodes $n,m$ of tree $T$ and any node $m'$ of tree $T'$, if $mRm'$ and $n$ is an ancesstor of $m$, then there is an ancestor $n'$ of $m'$ such that $nRn'$.
\end{enumerate}
\end{definition}
We say that two mechanisms are {\bf bisimular} if there is a bisimulation of them.
Note that we assume that bisimular mechanisms have the same set $Ag$ of agents and the same set $Alt$ of alternatives.

\begin{figure}[ht]
\begin{center}
\vspace{-2mm}
\scalebox{0.5}{\includegraphics{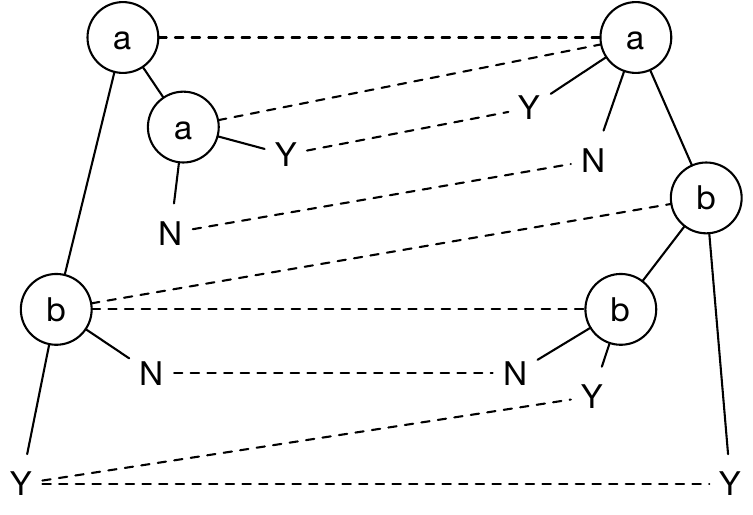}}
\caption{Dashed lines show a bisimulation of two decision-making mechanisms.}\label{bisimulation example}
\end{center}
\vspace{-4mm}
\end{figure}

Figure~\ref{bisimulation example} shows an example of a bisimulation of two mechanisms. 

\begin{definition}\label{total}
Bisimulation $R$ is {\bf\em total} if for each node $n$ there is a node $n'$ such that $nRn'$ and for each node $n'$ there is a node $n$ such that $nRn'$.   
\end{definition}
The bisimulation shown in Figure~\ref{bisimulation example} is total.

In this paper, by $height(n)$ of a node $n$ we mean the number of edges along the longest downward path from the node $n$ to a leaf node. 
For example, $height(n_2)=2$ and $height(n_5)=0$ in the mechanism shown in Figure~\ref{constitution figure}.

\subsection{Properties of a Bisimulation}

In this subsection, we assume a fixed pair of mechanisms $(Ag,Alt,T,act,\ell)$ and $(Ag,Alt,T',act',\ell')$ and a bisimulation $R$ of these mechanisms. The next key lemma establishes that bisimulation preserves the {\em strategic power} of each agent. Its proof is in the appendix. 

\begin{lemma}\label{win-win lemma}
$n\in win_a(S)$ iff $n'\in win_a(S)$ for 
any subset $S\subseteq Alt$ of alternatives and
any two nodes $n$ and $n'$ of trees $T$ and $T'$, respectively, such that $nRn'$. 
\end{lemma}

\begin{lemma}\label{dict iff dict}
If $nRn'$ and an agent $a$ is a dictator at node $n$, then  $a$ is a dictator at node $n'$.
\end{lemma}
\begin{proof}
Consider any alternative $Y\in Alt$. By Definition~\ref{dict at node}, it suffices to show that $n'\in win_a(\{Y\})$.

Note that $n\in win_a(\{Y\})$ by Definition~\ref{dict at node} and the assumption of the lemma that agent $a$ is a dictator at node~$n$. 
Therefore,
$n'\in win_a(\{Y\})$
by the assumption $nRn'$ of the lemma and Lemma~\ref{win-win lemma}. 
\end{proof}

\begin{lemma}\label{31-jul-d}
If bisimulation $R$ is total and every node $n$ of tree $T$ has a dictator at an ancestor of $n$, then every node $n'$ of tree $T'$ has a dictator at an ancestor of $n'$.
\end{lemma}
\begin{proof}
Consider any node $n'$ of tree $T'$. By Definition~\ref{total}, there is a node $n$ of tree $T$ such that
\begin{equation}\label{31-jul-c}
nRn'.    
\end{equation}
By the assumption of the lemma, there is a dictator agent $a$ at an ancestor $m$ of node $n$. By item~4 of Definition~\ref{def:bisimulation} and statement~\eqref{31-jul-c}, there is an ancestor $m'$ of node $n'$ such that $mRm'$. Then, $a$ is a dictator at node $m'$ by Lemma~\ref{dict iff dict}.   
\end{proof}

\begin{lemma}\label{resp iff resp}
For any two leaf nodes $n$ and $n'$, if $nRn'$ and an agent $a$ is responsible at $n$, then  $a$ is responsible at $n'$.
\end{lemma}
\begin{proof}
By Definition~\ref{responsible}, the assumption that $a$ is responsible at node $n$ implies that $n$ has an ancestor $m$ such that 
\begin{equation}\label{31-jul-a}
m\in win_a(Alt\setminus \{\ell(n)\}).    
\end{equation}
Thus, by item~4 of Definition~\ref{def:bisimulation} and the assumption $nRn'$ of the lemma, there is an ancestor $m'$ of node $n'$ such that $mRm'$. Hence, 
$
m'\in win_a(Alt\setminus \{\ell(n)\})
$
by Lemma~\ref{win-win lemma} and statement~\eqref{31-jul-a}.
Note that $\ell(n)=\ell'(n')$ by item~1 of Definition~\ref{def:bisimulation}. Then, 
$
m'\in win_a(Alt\setminus \{\ell'(n')\})
$.
Therefore, agent $a$ is responsible at leaf $n'$ by Definition~\ref{responsible}.
\end{proof}

We conclude this section with two important observations about arbitrary totally bisimular mechanisms.

\begin{lemma}\label{diffusion bisimilar}
If one of two totally bisimular mechanisms allows diffusion of responsibility, then so does the other.
\end{lemma}
\begin{proof}
The statement of the lemma follows from Lemma~\ref{resp iff resp} and Definition~\ref{total}.    
\end{proof}



\begin{lemma}\label{elected bisimilar}
If one of two totally bisimular mechanisms is an elected dictatorship, then so is the other. 
\end{lemma}
\begin{proof}
The statement of the lemma follows from Definition~\ref{elected dictatorship} and   Lemma~\ref{31-jul-d}.
\end{proof}

\subsection{Canonical Form}

\begin{definition}\label{canonical}
A mechanism is in {\bf\em canonical form} if there is no totally bisimular mechanism with a fewer number of nodes. 
\end{definition}
Note that neither of the two mechanisms depicted in Figure~\ref{bisimulation example} is in canonical form.

\begin{lemma}\label{parent-child lemma}
If a mechanism is in canonical form, then any node and its parent cannot be labeled with the same agent.  
\end{lemma}
\begin{proof}
Suppose that a parent node and its child are labeled with the same agent, see the diagram below (left).
\begin{center}
\scalebox{0.5}{\includegraphics{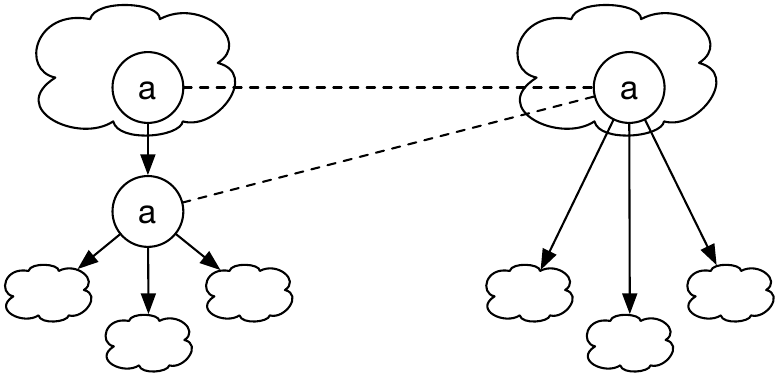}}
\vspace{-2mm}
\end{center}    
Consider a new mechanism (right) that ``collapses'' the parent and the child nodes into a single node. Let $R$ be the total bisimulation of the original and the new mechanism as shown by the dashed line on the diagram. We assume that all nodes not shown in the original mechanism are connected by the dashed lines to their clones in the new mechanism. Note that the new mechanism has a fewer number of nodes than the original mechanism. Therefore, by Definition~\ref{canonical}, the original mechanism is not in canonical form.
\end{proof}

\begin{lemma}\label{no twins lemma}
A node of a mechanism in canonical form cannot have two leaf children labeled with the same alternative.    
\end{lemma}
\begin{proof}
The proof of the lemma is again similar to the proof of Lemma~\ref{parent-child lemma}, using the diagram below.
\begin{center}
\scalebox{0.5}{\includegraphics{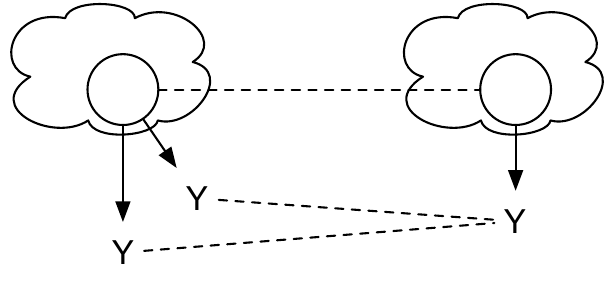}}
\end{center} 
\vspace{-4mm}
\end{proof}

\begin{lemma}\label{two children lemma}
If a mechanism is in canonical form, then each node of height 1 has at least two leaf children labeled with different alternatives.  
\end{lemma}
\begin{proof}
Consider any node of height 1 that does not have at least two children labeled with the same alternative. Thus, by Lemma~\ref{no twins lemma}, this node has only one leaf child. 
\begin{center}
\scalebox{0.5}{\includegraphics{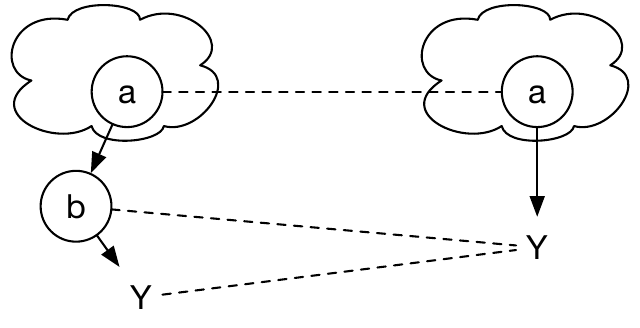}}
\vspace{-3mm}
\end{center}
The rest of the proof is similar to the proof of Lemma~\ref{parent-child lemma} using the diagram above.
\end{proof}

The proofs of the next three lemmas are in the appendix.

\begin{lemma}\label{no-single-child lemma}
A node of a mechanism in canonical form cannot have exactly one child.  
\end{lemma}

\begin{lemma}\label{no leaf-uncles lemma}
If a diffusion-free mechanism is in canonical form, then any child of a node of height 2 must have height~1.  
\end{lemma}

By $ChildAlt(n)$ we denote the set of alternatives used to label the leaf children of a node $n$. For example, 
$ChildAlt(n_2)=\{\text{Yes}\}$ and 
$ChildAlt(n_4)=\{\text{Yes},\text{No}\}$ for the mechanism depicted in Figure~\ref{constitution figure}. 

\begin{lemma}\label{childAlt lemma}
If a diffusion-free mechanism is in canonical form, and $n_1$ and $n_2$ are children of a node of height 2, then $ChildAlt(n_1)=ChildAlt(n_2)$.\end{lemma}




\begin{lemma}\label{two non-leaf children lemma}
If a diffusion-free mechanism is in canonical form, then any node of height 2 must have at least two children labeled with different agents.  
\end{lemma}
\begin{proof}
Consider any node $n$ of height 2. By the definition of the height, node $n$ must have at least one child. Thus, by Lemma~\ref{no-single-child lemma}, node $n$ must have at least two children, $n_1$ and $n_2$. By Lemma~\ref{no leaf-uncles lemma}, $height(n_1)=height(n_2)=1$.
It suffices to show $act(n_1)\neq act(n_2)$. Suppose $act(n_1)= act(n_2)=a$ for some agent $a$.

Lemma~\ref{childAlt lemma} implies that $ChildAlt(n_1)=ChildAlt(n_2)$. By Lemma~\ref{no twins lemma}, all children of node $n_1$ are labeled with different alternatives. The same is true about node $n_2$.
Hence, nodes $n_1$ 
and $n_2$ have exactly the same number of children labeled with the same set of alternatives (one child per alternative).

Consider a new mechanism that ``combines'' nodes $n_1$ and $n_2$ as well as their identically labeled children.  Let $R$ be the total bisimulation of the original and the new mechanism as shown by the dashed line on the diagram below:
\begin{center}
\scalebox{0.5}{\includegraphics{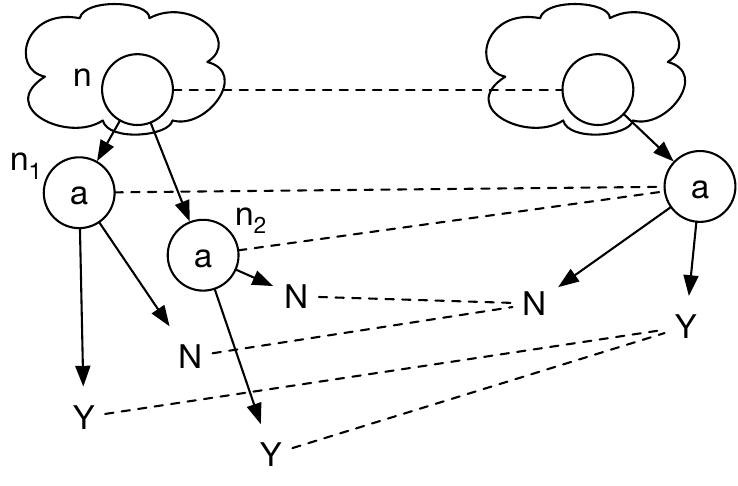}}
\vspace{-2mm}
\end{center}
We assume that all nodes not shown in the original mechanism are connected by the dashed lines to their clones in the new mechanism. Note that the new mechanism has fewer nodes than the original mechanism. Thus, by Definition~\ref{canonical}, the original mechanism is not in canonical form.
\end{proof}


\begin{lemma}\label{Jia's lemma}
In a diffusion-free mechanism  in canonical
form,
if $n$ is a parent of a non-leaf node $n_1$, then the tree rooted at node $n_1$ contains a node $m$ such that $height(m)=1$ and $act(m)\neq act(n)$.
\end{lemma}
\begin{proof}
If $height(n_1)=1$, then let $m$ be the node $n_1$. Note that $act(m)=act(n_1)\neq act(n)$ by Lemma~\ref{parent-child lemma} and the assumption that the mechanism is in canonical form. In what follows, we assume that 
\begin{equation}\label{7-aug-e}
height(n_1)\ge 2.   
\end{equation}
Let $n_2$ be the deepest leaf node in the subtree rooted at node $n_1$ as shown below:

\begin{center}
\vspace{-2mm}
\scalebox{0.5}{\includegraphics{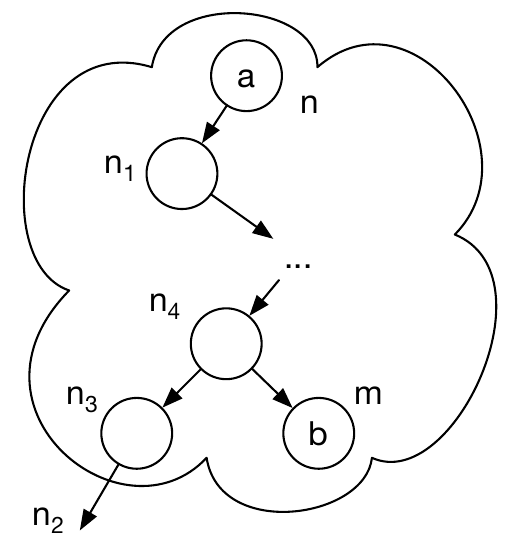}}
\vspace{-2mm}
\end{center}
Let $n_3$ be the parent of leaf $n_2$ and $n_4$ be the parent of node $n_3$. Note that $height(n_4)=2$ because $n_2$ is the deepest leaf node. Also, observe that node $n_4$ belongs to the subtree of node $n_1$ by the inequality~\eqref{7-aug-e}. By Lemma~\ref{two non-leaf children lemma}, node $n_4$ must have a child node $m$ such that $act(m)\neq act(n)$.
\end{proof}

\begin{lemma}\label{leaf's parent height 1}
If a diffusion-free mechanism is in canonical form, then the parent of each leaf node has height 1.    
\end{lemma}
\begin{proof}
Consider parent $n$ of a leaf node $n_1$. Towards the contradiction, suppose that $height(n)>1$. Then, node $n$ must have a child non-leaf node $n_2$. By Lemma~\ref{Jia's lemma}, the tree rooted at node $n_2$ contains a node $n_3$ such that $height(n_3)=1$ and \begin{equation}\label{2-aug-a}
act(n_3)\neq act(n).    
\end{equation}
By Lemma~\ref{two children lemma}, node $n_3$ has two leaf children labeled with different alternatives. Let $n_4$ be one of these leaf children such that $\ell(n_1)\neq \ell(n_4)$ and let $n_5$ be the other leaf child such that $\ell(n_4)\neq \ell(n_5)$, as shown:

\begin{center}
\vspace{0mm}
\scalebox{0.5}{\includegraphics{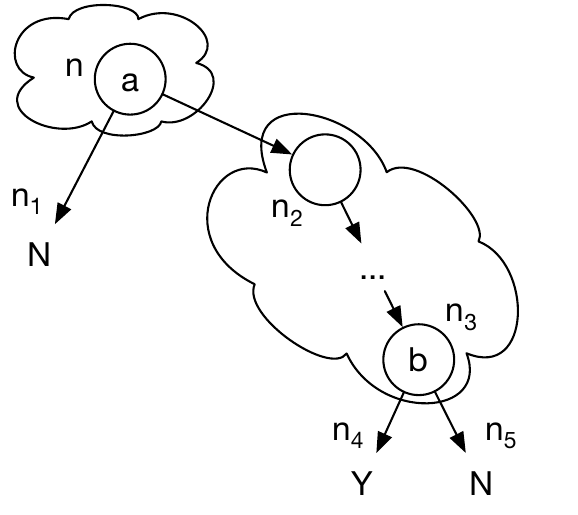}}
\vspace{-2mm}
\end{center}

First, note that agent $act(n)$ is responsible at node $n_4$. Indeed, because $\ell(n_1)\neq \ell(n_4)$, the agent had a strategy to prevent alternative $\ell(n_4)$ by transitioning the decision-making process from node $n$ to $n_1$.

Second, agent $act(n_3)$ is also responsible at node $n_4$. Indeed, because $\ell(n_4)\neq \ell(n_5)$, the agent had a strategy to prevent alternative $\ell(n_4)$ by transitioning the decision-making process from node $n_3$ to $n_5$.

Thus, agents $act(n)$ and $act(n_3)$ are both responsible at node $n_4$. Hence, $act(n)=act(n_3)$ by the assumption of the lemma that the mechanism is diffusion-free and Definition~\ref{diffusion}, which contradicts inequality~\eqref{2-aug-a}.
\end{proof}

By $TreeAlt(n)$ we refer to the set of all alternatives used to label the leaf nodes in the subtree of the mechanism rooted at node $n$. For the mechanism depicted in Figure~\ref{constitution figure}, we have
$TreeAlt(n_2)=\{\text{Yes},\text{No}\}$ and
$ChildAlt(n_2)=\{\text{Yes}\}$.

\begin{lemma}\label{child=tree lemma}
$ChildAlt(m)=TreeAlt(n)$, for any node $n$ of a diffusion-free mechanism in canonical form and any node $m$ of height~1 in the subtree rooted at $n$.
\end{lemma}
\begin{proof}
We prove the lemma by induction on $height(n)$. 

In the base case, $height(n)=1$.  Hence, $ChildAlt(n)=TreeAlt(n)$. Also, $n=m$ by the assumption $height(m)=1$ of the lemma. Therefore, $ChildAlt(m)=TreeAlt(n)$.

For the induction step, suppose that $height(n)>1$. Let $n_1$ be the child of node $n$ such that the subtree rooted at node $n_1$ contains node $m$. 
Note that 
$height(n_1)<height(n)$ because $n_1$ is a child of node $n$. 
Then, $ChildAlt(m)=TreeAlt(n_1)$ by the induction hypothesis. Thus, it suffices to show that $TreeAlt(n_1) =TreeAlt(n)$.

Suppose the opposite. Consider any alternative
\begin{equation}\label{7-aug-b}
Y\in TreeAlt(n)\setminus TreeAlt(n_1).    
\end{equation}
Then, there must exist a child $n_2$ of node $n$ such that 
\begin{equation}\label{7-aug-a}
   Y\in TreeAlt(n_2). 
\end{equation}
Observe that node $n_2$ cannot be a leaf node by the assumption $height(n)>1$ of the induction case and Lemma~\ref{leaf's parent height 1}. Hence, by Lemma~\ref{Jia's lemma}, the subtree rooted at $n_2$ contains a node $n_3$ such that $height(n_3)=1$ and 
\begin{equation}\label{7-aug-c}
act(n_3)\neq act(n),    
\end{equation}
as shown below:

\begin{center}
\scalebox{0.5}{\includegraphics{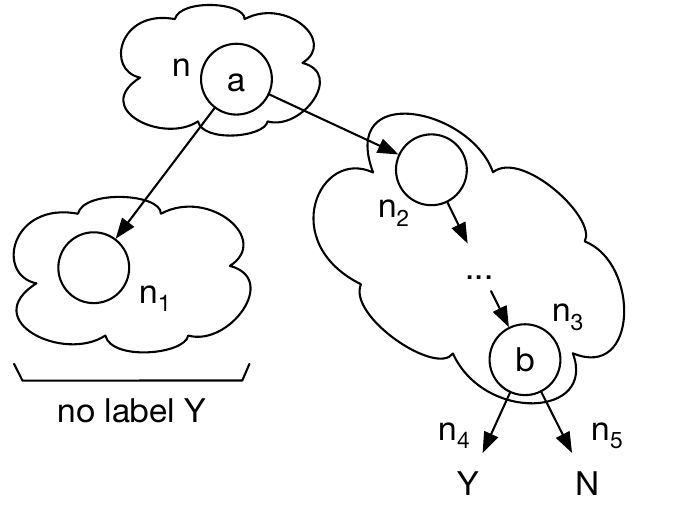}}
\vspace{-1mm}
\end{center}

Note $height(n_2)<height(n)$ because node $n_2$ is a child of node $n$. Thus, $ChildAlt(n_3)=TreeAlt(n_2)$
by the induction hypothesis. Hence, $Y\in ChildAlt(n_3)$ by statement~\eqref{7-aug-a}. Thus, 
node $n_3$ must have a child leaf node $n_4$ such that $\ell(n_4)=Y$. By Lemma~\ref{two children lemma}, node $n_3$ must also have another child leaf $n_5$ such that $\ell(n_5)\neq Y$.

First, note that agent $act(n)$ is responsible at node $n_4$. Indeed, because $Y\notin TreeAlt(n_1)$ by statement~\eqref{7-aug-b}, the agent had a strategy to prevent alternative $Y$ by transitioning the decision-making process from node $n$ to $n_1$.

Second, agent $act(n_3)$ is also responsible at node $n_4$. Indeed, because $\ell(n_5)\neq Y$, the agent had a strategy to prevent alternative $Y$ by transitioning the decision-making process from node $n_3$ to $n_5$.

Thus, agents $act(n)$ and $act(n_3)$ are both responsible at node $n_4$. Thus, $act(n)=act(n_3)$ by the assumption of the lemma that the mechanism is diffusion-free and Definition~\ref{diffusion}, which contradicts inequality~\eqref{7-aug-c}.
\end{proof}

\begin{lemma}\label{child= all alt}
If a diffusion-free mechanism is in canonical form, then $ChildAlt(n)=Alt$ for each node $n$ of height~1.    
\end{lemma}
\begin{proof}
Let $r$ be the root node of the mechanism. Then,  $ChildAlt(n)=AltTree(r)$ by Lemma~\ref{child=tree lemma}. At the same time, $AltTree(r)=Alt$ because labeling function $\ell$ is a surjection by item~5 of Definition~\ref{game}.     
\end{proof}

\subsection{Technical results}

We are now ready to prove the main result of this work.
\begin{theorem}\label{first main result}
Any diffusion-free mechanism is an elected dictatorship.    
\end{theorem}
\begin{proof}
Consider any diffusion-free mechanism $M$. Let $M'$ be any smallest (in terms of the number of nodes) mechanism totally bisimular to mechanism $M$. By Lemma~\ref{diffusion bisimilar} and Lemma~\ref{elected bisimilar}, it suffices to prove the statement of the theorem for mechanism $M'$. 

Note that mechanism $M'$ is in canonical form by Definition~\ref{canonical}. Consider any leaf node $m$. Node $m$ is not a root node by Lemma~\ref{no singleton mechanisms}.
Let $n$ be the parent of node $m$. Note that $height(n)=1$ by Lemma~\ref{leaf's parent height 1}. Thus, $Child(n)=Alt$ by Lemma~\ref{child= all alt}. Thus, $win_{act(n)}(Y)$ for each $Y\in Alt$ by Definition~\ref{df:winning set}.
Therefore, $act(n)$ is a dictator at node $n$ by Definition~\ref{dict at node}. Therefore, mechanism $m'$ 
is an elected dictatorship by Definition~\ref{elected dictatorship}.  
\end{proof}

The next observation about the special case of two-agent mechanisms can be proven directly or, as we do, derived from Theorem~\ref{first main result}. 
\begin{theorem}\label{two players main theorem}
Any diffusion-free two-agent mechanism is a dictatorship.   
\end{theorem}
\begin{proof}
Consider any two-agent diffusion-free mechanism. By Theorem~\ref{first main result}, this mechanism is an elected dictatorship. Thus, by Definition~\ref{elected dictatorship}, along any decision path, there is a node at which one of the agents is a dictator. 
Then, the mechanism can be viewed as a win-lose extensive form game, where the objective is to be the {\em first} to become the dictator at a node. 
It is a well-known fact in game theory~\cite[Theorem 3.5.1, p.91]{b18} that one of the players in a two-player win-lose extensive form game has a winning strategy. In our case, this means that one of the two agents has a strategy, at the root node, to reach a node at which the agent is a dictator. Recall that at such a node, the agent has a strategy to guarantee any alternative. Hence, the agent has a strategy at a root node to guarantee any alternative. Therefore, this agent is a dictator by Definition~\ref{dictator}.    
\end{proof}

\section{Conclusion}

In this paper, we have shown that in a consecutive collective decision-making mechanism under which each party meaningfully contributes to the decision process, it is unavoidable that in some cases the responsibility will be diffused between more than one agent. The precise statement of this result, given in Theorem~\ref{first main result} and Theorem~\ref{two players main theorem}, depends on the number of agents involved in the decision-making process.

\bibliography{naumov}

\clearpage

\begin{center}
{\LARGE\bf Technical Appendix}    
\end{center}

\noindent{\bf Lemma~\ref{win-win lemma}.} {\em 
$n\in win_a(S)$ iff $n'\in win_a(S)$ for 
any subset $S\subseteq Alt$ of alternatives and
any two nodes $n$ and $n'$ of trees $T$ and $T'$, respectively, such that $nRn'$. 
}
\begin{proof}
We prove the lemma by induction on $height(n)+height(n')$. In the base case, $height(n)=height(n')=0$. Thus, $n$ and $n'$ are the leaf nodes. Hence, $\ell(n)=\ell(n')$ by item~1 of Definition~\ref{def:bisimulation}. Then, $\ell(n)\in S$ iff $\ell'(n')\in S$. Therefore, $n\in win_a(S)$ iff $n'\in win_a(S)$ by Definition~\ref{df:winning set}.

Suppose that $height(n)+height(n')>0$. Without loss of generality, we can assume that 
\begin{equation}\label{26-jul-height>0}
height(n)>0.    
\end{equation}

\noindent
$(\Rightarrow):$ Assume that
\begin{equation}\label{26-jul-n-win}
n\in win_a(S)    
\end{equation}
and consider the following three cases separately:

\vspace{1mm}
\noindent{\em Case 1}: $act(n)=a$.
By Definition~\ref{df:winning set}, the assumptions
\eqref{26-jul-height>0} and \eqref{26-jul-n-win} imply that node $n$ has child $m$ such that 
\begin{equation}\label{26-jul-a}
m\in win_a(S).    
\end{equation}
Then, $n\stackrel{a}{\to}m$ by Definition~\ref{df:to} and the assumption $act(n)=a$ of the case.
Hence, by item~2 of Definition~\ref{def:bisimulation} and the assumption $nRn'$ of the lemma, there is a node $m'$ such that 
\begin{align}
&n'\stackrel{a}{\to}m',\label{26-jul-b}\\
&mRm'.\label{26-jul-c}
\end{align}
Note that $height(m)<height(n)$ because node $m$ is a child of node $n$. Also, $height(m')\le height(n')$ by Definition~\ref{df:to} and statement~\eqref{26-jul-b}.
Thus, $height(m)+height(m')<height(n)+height(n')$.
Hence, $m\in win_a(S)$ iff $m'\in win_a(S)$ by the induction hypothesis and statement~\eqref{26-jul-c}. Thus, $m'\in win_a(S)$ by statement~\eqref{26-jul-a}.
Therefore, $n'\in win_a(S)$ by item~\ref{up lemma} of Lemma~\ref{up down lemma}.

\vspace{1mm}
\noindent{\em Case 2}:
$act(n)=act(n')\neq a$. The assumption~\eqref{26-jul-height>0} implies that node $n$ is not a leaf. Thus, node $n'$ is also not a leaf by the assumption $act(n)=act(n')$ of the case and item~4 of Definition~\ref{game}. 

Then, to show that $n'\in win_a(S)$, consider any child $m'$ of node $n'$. By item~3 of Definition~\ref{df:winning set} and the assumption $act(n')\neq a$ of the case, it suffices to show that 
\begin{equation}
m'\in win_a(S).\label{26-jul-d}
\end{equation}
Note that, 
$n'\stackrel{act(n)}{\to} m'$
by the assumption $act(n)=act(n')$ of the case and Definition~\ref{df:to}. 
Hence, by item~3 of Definition~\ref{def:bisimulation} and the assumption $nRn'$ of the lemma, there is node $m$ such that 
\begin{align}
&n\stackrel{act(n)}{\to} m,\label{26-jul-e}\\
&mRm'.\label{26-jul-f}
\end{align}
Then, by item~\ref{down lemma} of Lemma~\ref{up down lemma},
\begin{equation}
m\in win_a(S).\label{26-jul-g}    
\end{equation}

Note that $height(m')<height(n')$ because node $m'$ is a child of node $n'$. Also, $height(m)\le height(n)$ by Definition~\ref{df:to} and statement~\eqref{26-jul-e}.
Thus, $height(m)+height(m')<height(n)+height(n')$.
Hence, by statement~\eqref{26-jul-f} and the induction hypothesis,
$m\in win_a(S)$ iff $m'\in win_a(S)$. Therefore, statement~\eqref{26-jul-g} implies statement~\eqref{26-jul-d}.

\vspace{1mm}
\noindent{\em Case 3}: $act(n)\neq a$ and $act(n)\neq act'(n')$. The last inequality includes the situation when $act'(n')$ is not defined because $n'$ is a leaf node.

The assumption~\eqref{26-jul-height>0} implies that node $n$ is not a leaf node. Consider any child $m$ of this node. Then, by Definition~\ref{df:to},
\begin{equation}\label{26-jul-k}
n\stackrel{act(n)}{\to} m    
\end{equation}
and, by item~3 of Definition~\ref{df:winning set}, the assumption~$act(n)\neq a$ of the case and statement~\eqref{26-jul-n-win},
\begin{equation}
m\in win_a(S).  \label{26-jul-l}  
\end{equation}

By item~2 of Definition~\ref{def:bisimulation}, statement~\eqref{26-jul-k} and the assumption $nRn'$ of the lemma, there is a node $m'$ such 
that 
\begin{align}
&n'\stackrel{act(n)}{\to} m',\label{26-jul-m}\\
&mRm'.\label{26-jul-n}
\end{align}
By Lemma~\ref{no go lemma} and the assumption $act(n)\neq act'(n')$ of the case, statement~\eqref{26-jul-m} implies $n'=m'$. Hence, by statement~\eqref{26-jul-n},
\begin{equation}\label{26-jul-o}
mRn'.    
\end{equation}
Note that $height(m)+height(m')<height(n)+height(n')$ because $m$ is a child of node $n$ and $m'=n'$.
Hence, 
$m\in win_a(S)$ iff $n'\in win_a(S)$,
by statement~\eqref{26-jul-o} and the induction hypothesis.
Therefore, $n'\in win_a(S)$ by statement~\eqref{26-jul-l}.

\noindent
$(\Leftarrow)$ Assume that
\begin{equation}\label{29-jul-a-win}
n'\in win_a(S)    
\end{equation}
and consider the following three cases separately:

\vspace{1mm}
\noindent{\em Case 1:} $act(n)=a\neq act'(n')$. This case includes the situation when $n'$ is a leaf node and $act'(n')$ is not defined. Inequality~\eqref{26-jul-height>0} implies that node $n$ has at least one child $m$. Then, by the assumption $act(n)=a$ of the case,
\begin{equation}\label{29-jul-c}
  n\stackrel{a}{\to} m.  
\end{equation}
Hence, by the assumption $nRn'$ of the lemma and item~2 of Definition~\ref{def:bisimulation}, there is a node $m'$ such that $n'\stackrel{a}{\to} m'$ and $mRm'$. Then, $n'=m'$ by Lemma~\ref{no go lemma} and the assumption $a\neq act'(n')$ of the case. Thus, 
\begin{equation}\label{29-jul-b}
    mRn'.
\end{equation}
Note $height(m)+height(n')<height(n)+height(n')$ because $m$ is a child of node $n$.
Hence, 
$m\in win_a(S)$ iff $n'\in win_a(S)$,
by statement~\eqref{29-jul-b} and the induction hypothesis. Thus, $m\in win_a(S)$ by statement~\eqref{29-jul-a-win}. Therefore, $n\in win_a(S)$ by item~\ref{up lemma} of Lemma~\ref{up down lemma} and statement~\eqref{29-jul-c}.

\vspace{1mm}
\noindent{\em Case 2:} $act(n)=a= act'(n')$. Assumption $act'(n')=a$ implies that node $n'$ is not a leaf. Thus, by Definition~\ref{df:winning set} and the assumption $act'(n')=a$, node $n'$ must have at least one child $m'$ such that
\begin{equation}\label{29-jul-d}
   m'\in win_a(S). 
\end{equation}
Note that 
$n'\stackrel{a}{\to}m'$ by the assumption $act'(n')=a$. Hence, by the assumption $nRn'$ of the lemma and item~3 of Definition~\ref{def:bisimulation}, there must exist a node $m$ such that 
\begin{align}
&n\stackrel{a}{\to}m,\label{29-jul-e}\\
&mRm'.\label{29-jul-f}
\end{align}
Note that $height(m')<height(n')$ because node $m'$ is a child of node $n'$. Also, $height(m)\le height(n)$ by Definition~\ref{df:to} and statement~\eqref{29-jul-e}.
Thus, $height(m)+height(m')<height(n)+height(n')$.
Hence, by statement~\eqref{29-jul-f} and the induction hypothesis,
$m\in win_a(S)$ iff $m'\in win_a(S)$.
Then, $m\in win_a(S)$ by statement~\eqref{29-jul-d}.
Therefore,  item~\ref{up lemma} of Lemma~\ref{up down lemma} and statement~\eqref{29-jul-e} imply that $n\in win_a(S)$.

\vspace{1mm}
\noindent{\em Case 3:} $act(n)\neq a$. Note that $act(n)$ is defined because, by assumption~\eqref{26-jul-height>0}, node $n$ is not a leaf. 

Consider any child $m$ of node $n$ by Definition~\ref{df:winning set}, it suffices to show that $m\in win_a(S)$. Indeed, note that $n\stackrel{act(n)}{\to}m$ by Definition~\ref{df:to}. Hence, by item~2 of Definition~\ref{def:bisimulation} and assumption $nRn'$ of the lemma, there exists a node $m'$ such that 
\begin{align}
&n'\stackrel{act(n)}{\to}m',\label{29-jul-g}\\
&mRm'.\label{29-jul-h}    
\end{align}
Thus, statement~\eqref{29-jul-a-win} and the assumption $act(n)\neq a$ of the case, by item~\ref{down lemma} of Lemma~\ref{up down lemma}, imply
\begin{equation}\label{29-jul-i} 
m'\in win_a(S).    
\end{equation}

Note that $height(m)<height(n)$ because node $m$ is a child of node $n$. Also, $height(m')\le height(n')$ by Definition~\ref{df:to} and statement~\eqref{29-jul-g}.
Thus, $height(m)+height(m')<height(n)+height(n')$.
Hence, by statement~\eqref{29-jul-h} and the induction hypothesis,
$m\in win_a(S)$ iff $m'\in win_a(S)$.
Then, $m\in win_a(S)$ by statement~\eqref{29-jul-i}.
\end{proof}

\noindent{\bf Lemma~\ref{no-single-child lemma}.} {\em
A node of a mechanism in canonical form cannot have exactly one child.  
}
\begin{proof}
By Lemma~\ref{two children lemma}, any node with exactly one child cannot have height 1. Thus, the single child of this node cannot be a leaf node.  The rest of the proof of this lemma is similar to the proof of Lemma~\ref{parent-child lemma} using the diagram below.
\begin{center}
\scalebox{0.5}{\includegraphics{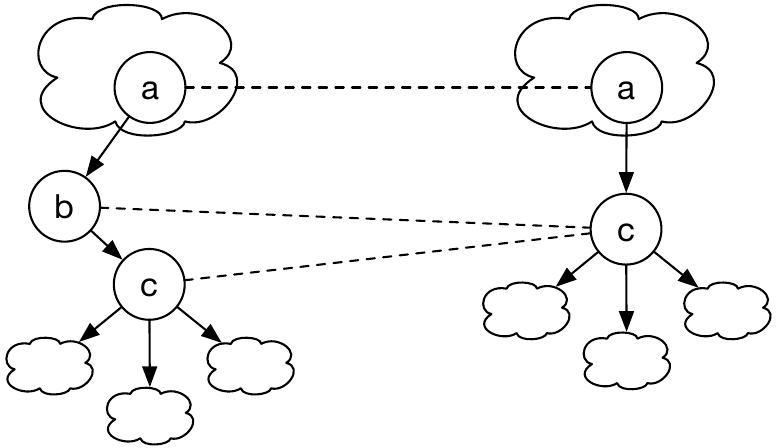}}
\end{center}   
\end{proof}

\noindent{\bf Lemma~\ref{no leaf-uncles lemma}.} {\em 
If a diffusion-free mechanism is in canonical form, then any child of a node of height 2 must have height~1.  
}

\begin{proof}
Consider any node $n$ of height 2 and any child $m$ of this node. Note that $height(m)\le height(n)-1=1$. Thus, it suffices to show that $m$ is not a leaf node. Suppose the opposite. 

Since $height(n)=2$, node $n$ must have at least one child $m_1$ of height 1. By Lemma~\ref{two children lemma}, node $m_1$ has at least two children labeled with different alternatives. Let $m_2$ be child of node $m_1$ such that $\ell(m)\neq \ell(m_2)$ and $m_3$ be another child of $m_1$ such $\ell(m_3)\neq \ell(m_2)$, as shown in the diagram below:
\begin{center}
\scalebox{0.5}{\includegraphics{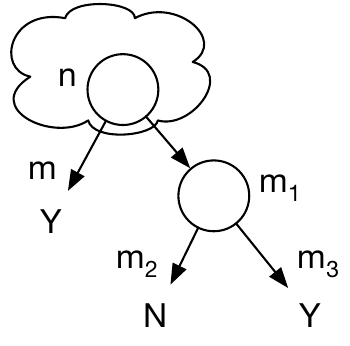}}
\vspace{-2mm}
\end{center} 

First, note that agent $act(n)$ is responsible at node $m_2$. Indeed, because $\ell(m)\neq \ell(m_2)$, the agent had a strategy to prevent alternative $\ell(m_2)$ by transitioning the decision-making process from node $n$ to $m$.

Second, agent $act(m_1)$ is also responsible at node $m_2$. Indeed, because $\ell(m_3)\neq \ell(m_2)$, the agent had a strategy to prevent alternative $\ell(m_2)$ by transitioning the decision-making process from node $m_1$ to $m_3$.

Thus, agents $act(n)$ and $act(m_1)$ are both responsible at node $m_2$. Thus, $act(n)=act(m_1)$ by the assumption of the lemma that the mechanism is diffusion-free and Definition~\ref{diffusion}. The last statement contradicts Lemma~\ref{parent-child lemma}.
\end{proof}

\noindent{\bf Lemma~\ref{childAlt lemma}.} {\em
If a diffusion-free mechanism is in canonical form, and $n_1$ and $n_2$ are children of a node of height 2, then $ChildAlt(n_1)=ChildAlt(n_2)$.
}
\begin{proof}
By Lemma~\ref{no leaf-uncles lemma}, the assumptions of the lemma imply $height(n_1)=height(n_2)=1$, see the diagram below.
\begin{center}
\scalebox{0.5}{\includegraphics{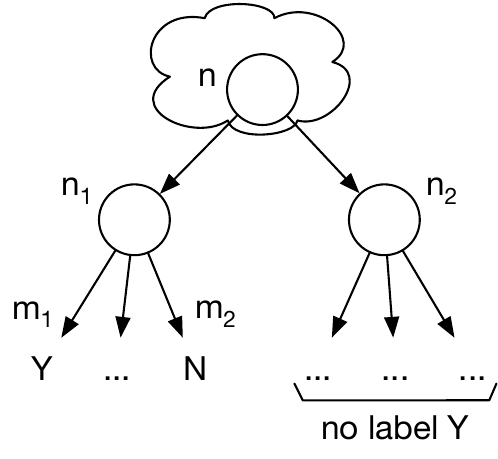}}
\vspace{-1mm}
\end{center}

Suppose $ChildAlt(n_1)\neq ChildAlt(n_2)$. Without loss of generality, let $Y\in ChildAlt(n_1)\setminus ChildAlt(n_2)$ for some $Y\in Alt$. Thus, node $n_1$ has a leaf child $m_1$ labeled with the alternative $Y\notin ChildAlt(n_2)$. Then, agent $act(n)$ is responsible at leaf node $m_1$ because the agent had a strategy to prevent alternative $Y$ by transitioning the decision-making process from node $n$ to node $n_2$.

At the same time, the assumption of the lemma that the mechanism is in canonical form, by Lemma~\ref{two children lemma}, implies that node $n_1$ has another leaf child, $m_2$, such that $\ell(m_1)\neq \ell(m_2)$. Then, agent $act(n_1)$ is responsible at leaf node $m_1$ because the agent had a strategy to prevent alternative $Y$ by transitioning the decision-making process from node $n_1$ to leaf node $m_2$.

Hence, agents $act(n)$ and $act(n_1)$ are both responsible at leaf node $m_1$. Then, $act(n)=act(n_1)$ by the assumption of the lemma that the mechanism is diffusion-free and Definition~\ref{diffusion}. The last statement contracts Lemma~\ref{parent-child lemma} because the mechanism is in canonical form.
\end{proof}

\end{document}